\documentclass[letterpaper, 10 pt, conference]{ieeeconf}  

\IEEEoverridecommandlockouts                              
\usepackage{graphicx} 
\usepackage{amsmath} 
\usepackage{amssymb}  
\usepackage{algorithm}
\usepackage[noend]{algpseudocode}
\usepackage{dsfont}
\newcommand{\Given}[1]{\Statex \textbf{Given:} #1}
\newcommand{\Input}[1]{\Statex \textbf{Input:} #1}
\newcommand{\Output}[1]{\Statex \textbf{Output:} #1}

\usepackage{amsthm}
\usepackage{multirow}
\usepackage{array}
\usepackage[hidelinks]{hyperref}

\theoremstyle{plain}
\newtheorem{theorem}{Theorem}
\newtheorem{proposition}{Proposition}
\theoremstyle{remark}
\newtheorem{remark}{Remark}
\theoremstyle{definition}
\newtheorem{definition}{Definition}

\title{\LARGE \bf
Sampling-Based Safety Filter\\with Probabilistic Restrictiveness Guarantee
}

\author{Junyoung Park, Hyeontae Sung and Heejin Ahn
\thanks{This work was partly supported by the IITP (Institute of Information \& Communications Technology Planning \& Evaluation)-ITRC (Information Technology Research Center) grant funded by the Korea government (MSIT) (IITP-2026-RS-2023-00259991, 50\%), and Basic Science Research Program through the National Research Foundation of Korea (NRF) funded by the Ministry of Education (RS-2025-25407010, 50\%).}
\thanks{The authors are with the School of Electrical Engineering,
        Korea Advanced Institute of Science and Technology (KAIST),
        Daejeon, Republic of Korea
        {\tt\small \{junyoung766, hyeontae.sung, heejin.ahn\}@kaist.ac.kr}}%
}

\begin{document}

\maketitle
\thispagestyle{empty}
\pagestyle{empty}

\begin{abstract}

Ensuring safety is a critical requirement for autonomous systems, yet providing formal guarantees for nominal controllers remains a significant challenge. In this paper, we propose a modular sampling-based safety filter to ensure the safety of arbitrary nominal control inputs. At each timestep, the filter evaluates the safety of the nominal input by leveraging control sequence samples generated via Stein Variational Model Predictive Control (SV-MPC). This approach approximates a safety-conditioned posterior distribution over control sequences, enabling the filter to effectively capture multimodal safe regions in complex, non-convex environments.
The filter guarantees safety by overriding the nominal input when all sampled control sequence candidates are deemed unsafe. By leveraging the scenario approach, the proposed method provides a probabilistic guarantee on its restrictiveness. We validate the filter through collision avoidance tasks in both single- and multi-vehicle settings, demonstrating its efficacy in navigating cluttered environments where nominal controllers may fail.

\end{abstract}

\section{INTRODUCTION}

In safety-critical control systems such as autonomous vehicles, control inputs must be generated with rigorous safety assurances to prevent catastrophic failures. However, nominal controllers, often designed via reinforcement learning or heuristic methods, are typically optimized for performance rather than safety. Providing formal guarantees for such policies remains a significant challenge, necessitating an additional supervisory mechanism to ensure safe operation. Safety filters offer a practical, modular solution by monitoring the nominal controller’s inputs and intervening only when necessary to enforce safety constraints.

Safety filters can be categorized by their certification mechanism.
Some rely on precomputed certificates, such as Hamilton–Jacobi (HJ) reachability~\cite{bansal2017hamilton, borquez2024safety} or control barrier functions (CBF)~\cite{ames2016control, wang2017safety}. While these approaches are computationally efficient at runtime, HJ value function computation scales poorly with system dimensions, and CBF synthesis remains a nontrivial challenge. Alternatively, methods based on model predictive control (MPC)~\cite{zeng2021safety, kopp2024data} verify safety online by predicting future trajectories over a finite horizon. While these eliminate the need for explicit certificate construction, gradient-based formulations often struggle with non-smooth objectives, non-convex constraints, and numerical infeasibility.

In this paper, we propose a sampling-based safety filter that operates entirely online. Our method is inspired by sampling-based MPC algorithms, such as model predictive path integral (MPPI)~\cite{williams2017information}, the cross-entropy method (CEM)~\cite{kobilarov2012cross}, and Stein variational MPC (SV-MPC)~\cite{lambert2020stein}, which are particularly effective in non-differentiable or non-convex settings. Our method evaluates the nominal input by approximating a safety-conditioned posterior distribution over future control sequences and sampling from it. These sampled sequences are propagated through discrete-time dynamics and evaluated via a level function encoding the safety constraints. By maintaining a candidate safe sequence as a backup at each timestep, the filter intervenes whenever all sampled trajectories fail to satisfy the safety requirements.

Our approach differs from prior work in two key respects. First, we leverage SV-MPC to capture the multimodal structure of safe control distributions. In contrast, existing sampling-based safety filters~\cite{park2025safety,feng2025words} typically draw control sequences from a unimodal Gaussian distribution, which limits their ability to effectively represent disjoint or non-convex safe regions. Second, we provide a probabilistic guarantee on the filter’s restrictiveness by leveraging the scenario approach~\cite{campi2009scenario}.
In particular, the filter overrides the nominal input only when it is unlikely that a safe future control sequence exists under the sampling distribution.
To the best of our knowledge, this is the first work to provide a formal probabilistic guarantee on the restrictiveness of a sampling-based safety filter under a finite number of samples. 

The contributions of this paper are as follows:
\begin{itemize}
    \item We propose a sampling-based safety filter that provides a probabilistic guarantee on its restrictiveness using the scenario approach.
    \item We incorporate Stein variational model predictive control (SV-MPC) into the filter to represent the multimodal safety-conditioned posterior distribution over control sequences.
    \item We validate the proposed filter on collision avoidance tasks in both single- and multi-vehicle scenarios with different nominal controllers.
\end{itemize}

This paper is organized as follows. Section~\ref{sec:prob} formulates the safety filtering problem, and Section~\ref{sec:ssf} presents the proposed safety filter, including the filter algorithm, the construction of the sampling distribution, and the theoretical properties of the filter.
Section~\ref{sec:exp} demonstrates the proposed method in single- and multi-vehicle scenarios.
Finally, Section~\ref{sec:conclusion} concludes the paper.
\section{PROBLEM STATEMENT} \label{sec:prob}

Consider a discrete-time dynamical system $\mathbf{x}_{t+1}=f(\mathbf{x}_t,\mathbf{u}_t)$, where $\mathbf{x}_t\in\mathcal{X}$ denotes the state and $\mathbf{u}_t\in\mathcal{U}$ denotes the control input at time $t$. Let $l:\mathcal{X}\rightarrow\mathbb{R}$ be a level function whose subzero level set defines the failure set
\begin{equation} \label{eq:failure_set}
    \mathcal{L}:=\{\mathbf{x}:l(\mathbf{x})\le0\}.
\end{equation}
The failure set represents unsafe states such as collisions or constraint violations. Therefore, ensuring \textit{safety} means that the system state remains outside $\mathcal{L}$.

For a finite horizon $H$, let $U_t:=(\mathbf{u}_t,\cdots,\mathbf{u}_{t+H-1})$ denote the control sequence starting at time $t$, and $X_t:=(\mathbf{x}_t,\cdots,\mathbf{x}_{t+H})$ the corresponding state sequence. Since the dynamics are deterministic, $X_t$ is uniquely determined by $\mathbf{x}_t$ and $U_t$. We define an auxiliary binary random variable $\mathcal{O}_{\tau_t}\in\{0,1\}$ for the state-action trajectory $\tau_t=(X_t,U_t)$ as
\begin{equation}\label{eq:safety}
    \begin{aligned}
    \mathcal{O}_{\tau_t} =
    \begin{cases}
        1,\quad &\mathbf{x}_{t+k}\notin\mathcal{L},\;\forall k\in[0,H],\\
        0,\quad &\text{otherwise}.
    \end{cases}
    \end{aligned}
\end{equation}
This means that $\mathcal{O}_{\tau_t} = 1$ if the trajectory $\tau_t$ is safe over the horizon $H$, and $\mathcal{O}_{\tau_t} = 0$ otherwise.

At each timestep $t$, a nominal controller generates a control input $\mathbf{u}^\mathrm{nom}_t\in \mathcal{U}$, which may be designed for arbitrary objectives and may not explicitly account for safety constraints. 
Therefore, directly applying $\mathbf{u}^\mathrm{nom}_t$ does not in general guarantee avoidance of $\mathcal{L}$.

In this paper, we design a sampling-based safety filter that samples $N$ candidate control sequences at each time step to assess the safety of the nominal input. Specifically, given the current state $\mathbf{x}_t$ and a nominal input $\mathbf{u}_t^{\mathrm{nom}}$, we first compute the next state $\mathbf{x}_{t+1} = f(\mathbf{x}_t, \mathbf{u}_t^{\mathrm{nom}})$. We then sample $N$ candidate control sequences $U_{t+1}^i$, $i \in [1, N]$, and propagate each sequence from $\mathbf{x}_{t+1}$ to generate a trajectory $\tau_{t+1}^i$. The safety filter outputs the control input $\mathbf{u}_t^{\mathrm{safe}}$ as
\begin{equation}\label{eq:u_safe}
\mathbf{u}_t^{\mathrm{safe}} =
\begin{cases}
\mathbf{u}_t^{\mathrm{nom}}, & \exists\, i\in[1,N]\;\text{s.t.}\;\mathcal{O}_{\tau^i_{t+1}}=1, \\
\mathbf{u}_t^{\mathrm{backup}}, & \text{otherwise},
\end{cases}
\end{equation}
where $\mathbf{u}_t^{\mathrm{backup}}$ denotes a known safe control input. 
The $N$ control sequences are sampled from a distribution $\tilde q_{t+1}(U)$.
The construction of $\tilde q_{t+1}(U)$ is described in Section~\ref{subsec_svmpc}.

We aim to design a safety filter \eqref{eq:u_safe} that satisfies the following two properties:
\begin{itemize}
    \item \textbf{Safety}: 
    For any $\mathbf{x}_{t}\notin\mathcal{L}$, there exists a control sequence $U_{t+1}$ such that from $\mathbf{x}_{t+1}=f(\mathbf{x}_t, \mathbf{u}_t^{\mathrm{safe}})$, the trajectory $\tau_{t+1}$ is safe, i.e., $\mathcal{O}_{\tau_{t+1}}=1$.
    If this holds, we say that \textit{the filter guarantees safety at time $t$}.
    \item \textbf{$\epsilon$-restrictiveness}: The filter overrides the nominal input only when it is unlikely that a safe future control sequence exists. Formally, an intervention implies that the probability of drawing a safe sequence from the  sampling distribution $\tilde q_{t+1}(U)$ is bounded by $\epsilon$, that is,
    \begin{equation*}
    \text{if } \mathbf{u}^{\mathrm{safe}}_{t} \neq \mathbf{u}_{t}^{\mathrm{nom}}, \text{ then } \Pr_{U_{t+1} \sim \tilde q_{t+1}}(\mathcal{O}_{\tau_{t+1}}=1) \le \epsilon.
    \end{equation*}
\end{itemize}


In the definition of $\epsilon$-restrictiveness, $\epsilon\in(0,1)$ denotes a restrictiveness parameter, and the uncertainty in $\mathcal{O}_{\tau_{t+1}}$ arises from the sampled control sequence $U_{t+1}\sim \tilde q_{t+1}$.
Smaller $\epsilon$ leads to less restrictive interventions.
The least-restrictive case corresponds to $\epsilon=0$, where intervention occurs only when no safe trajectory exists after applying $\mathbf{u}^\mathrm{nom}_t$.

\section{Sampling-Based Safety Filter} \label{sec:ssf}

In this section, we present the proposed sampling-based safety filter in detail. We first describe the overall algorithm and the construction of the sampling distribution. We then establish the safety guarantee and restrictiveness of the proposed filter.

\subsection{Safety Filter Algorithm}

\begin{algorithm}
\caption{Sampling-Based Safety Filter} \label{alg:1}
\begin{algorithmic}[1]
\Given{Trajectory cost function $C(\cdot)$}
\Input{Current state $\mathbf{x}_t$, nominal input $\mathbf{u}^\mathrm{nom}_t$, safe control sequence $U^\mathrm{safe}_t=(\bar{\mathbf{u}}_{t},\cdots, \bar{\mathbf{u}}_{t+H-1})$}
\Output{Safe input $\mathbf{u}^\mathrm{safe}_t$, safe control sequence $U^\mathrm{safe}_{t+1}$}
\State $\mathbf{x}_{t+1} \leftarrow f(\mathbf{x}_t, \mathbf{u}^\mathrm{nom}_t)$
\For{$i \leftarrow 1$ to $N$ \textbf{in parallel}}
    \State Sample $U^i_{t+1} \sim\tilde q_{t+1} \quad \triangleright$  Algorithm~\ref{alg:2}
    \State $X^i_{t+1} \leftarrow \text{Rollout}(\mathbf{x}_{t+1}, f(\cdot), U^i_{t+1})$
\EndFor
\State $i^* \leftarrow {\text{arg\,min}}_i\;C(\tau^i_{t+1})$
\If{$C(\tau^{i^*}_{t+1})<0$}
    \State $\mathbf{u}^\mathrm{safe}_t \leftarrow \mathbf{u}^\mathrm{nom}_t,\; U^\mathrm{safe}_{t+1} \leftarrow U^{i^*}_{t+1}$
\Else
    \State $\mathbf{u}^\mathrm{safe}_t \leftarrow \bar{\mathbf{u}}_{t},\; U^\mathrm{safe}_{t+1} \leftarrow \text{Shift}(U^\mathrm{safe}_t)$
\EndIf
\State \Return $\mathbf{u}^\mathrm{safe}_t, U^\mathrm{safe}_{t+1}$
\end{algorithmic}
\end{algorithm}

Our safety filtering algorithm is summarized in Algorithm~\ref{alg:1}. 
At each timestep $t$, the algorithm first predicts the next state $\mathbf{x}_{t+1}$ using the nominal control input $\mathbf{u}^\mathrm{nom}_t$ (line 1). 
It then samples $N$ control input sequences $U^i_{t+1}\;(i=1,\cdots,N)$ of length $H$ and rolls them out from $\mathbf{x}_{t+1}$ to generate $N$ corresponding state trajectories $X^i_{t+1}\;(i=1,\cdots,N)$ (lines 2--4). 
The safety of each trajectory is evaluated using the cost function
\begin{equation} \label{eq:cost}
C(\tau^i_{t+1}):=\max_{k\in[0,H]}\left\{-l(\mathbf{x}^i_{t+1+k})\right\},
\end{equation}
where $C(\tau_{t+1}^i)<0$ for safe trajectories and $C(\tau_{t+1}^i)\ge0$ otherwise.
In our setting, the cost depends only on the state trajectory $X_{t+1}^i=(\mathbf{x}_{t+1}^i,\cdots,\mathbf{x}_{t+H+1}^i)$; however, we use the state-action trajectory notation $\tau$ for generality.

The filter then selects the control sequence with the minimum cost among all samples (line 5).
If the minimum cost is negative, the nominal input $\mathbf{u}^\mathrm{nom}_t$ is applied without intervention, and $U^\mathrm{safe}_{t+1}$ is updated to the minimum-cost sample $U^{i^*}_{t+1}$ (lines 6--7). 
Otherwise, the filter applies the backup safe input $\mathbf{u}^\mathrm{safe}_t$ from the previously stored safe sequence $U^\mathrm{safe}_t$, and constructs $U^\mathrm{safe}_{t+1}$ by shifting $U^\mathrm{safe}_t$ forward by one step (lines 8--9). 
The last input of the shifted sequence is chosen to drive the system toward, or keep it inside, a safe control invariant set. Details are provided in Section~\ref{subsec_safety-guarantee}. 
The applied safe input $\mathbf{u}^\mathrm{safe}_t$ is then executed, and the updated safe sequence $U^\mathrm{safe}_{t+1}$ is stored for the next timestep (line 10).

\subsection{Sampling Distribution}\label{subsec_svmpc}

In line 3 of Algorithm~\ref{alg:1}, we draw samples to evaluate the safety of the nominal input $\mathbf{u}^\mathrm{nom}_t$. 
The efficacy of the safety filter depends on its ability to thoroughly explore the safe control space. Because the filter intervenes only when it fails to identify at least one safe candidate, a more comprehensive search directly reduces the frequency of unnecessary interventions. Furthermore, any safe sequence identified during this exploration serves as a potential safe backup $U^\mathrm{safe}_{t+1}$ for fallback control in subsequent timesteps.

To this end, we define a safety-conditioned posterior $p_{t+1}(U|\mathcal{O}_{\tau_{t+1}}=1)$. In this framework, we treat the search for safe control sequences as a Bayesian inference problem. 
By Bayes' rule, the safety-conditioned posterior is given by
\begin{equation} \label{eq:posterior}
    \begin{aligned}
        p_{t+1}(U|\mathcal{O}_{\tau_{t+1}}=1)=\frac{p_{t+1}(\mathcal{O}_{\tau_{t+1}}=1|U)p_{t+1}(U)}{p_{t+1}(\mathcal{O}_{\tau_{t+1}}=1)}.
    \end{aligned}
\end{equation}
For brevity, we write $\mathcal{O}_{\tau_{t+1}}$ instead of $\mathcal{O}_{\tau_{t+1}}=1$ whenever the meaning is clear.

Since the posterior is generally intractable, we approximate $p_{t+1}(U|\mathcal{O}_{\tau_{t+1}})$ using variational inference. 
Specifically, we seek a distribution $q^*(U)$ in a tractable family $Q$ that minimizes the Kullback-Leibler (KL) divergence to the target posterior:
$$q^*=\arg\min_{q\in Q}D_{KL}(q(U)||p_{t+1}(U|\mathcal{O}_{\tau_{t+1}})).$$
Using \eqref{eq:posterior}, this optimization can be rewritten as
\begin{equation} \label{eq:q_star}
    \begin{aligned}
        q^*=\arg\min_{q\in Q}\bigl\{-&\mathbb{E}_q\left[\log p_{t+1}(\mathcal{O}_{\tau_{t+1}}|U)\right]\\
        &+D_{KL}(q(U)\Vert p_{t+1}(U))\bigr\}.
    \end{aligned}
\end{equation}
To evaluate the safety-conditioned likelihood, we define a non-negative cost-likelihood function $L(\tau_{t+1}) \propto p_{t+1}(\mathcal{O}_{\tau_{t+1}} \mid U)$, which maps the predicted trajectory $\tau_{t+1}$ to a safety-informed weight. 
A standard formulation for this likelihood is $L(\tau_{t+1}) = \exp(-\alpha C(\tau_{t+1}))$ where $\alpha > 0$ is a temperature parameter and $C(\tau_{t+1})$ is the trajectory cost. 
This exponential form ensures that trajectories with higher costs, those nearing or entering the failure set $\mathcal{L}$, are exponentially less likely to be represented in the posterior. 
With this definition, \eqref{eq:q_star} becomes
\begin{equation} \label{eq:q_star_final}
    q^*=\arg\min_{q\in Q}\bigl\{-\mathbb{E}_q\left[\log L(\tau_{t+1})\right]
+D_{KL}(q(U)\Vert p_{t+1}(U))\bigr\},
\end{equation}
which is the variational formulation used in sampling-based MPC.

\begin{algorithm}
\caption{SV-MPC~\cite{lambert2020stein} for Safety Filtering} \label{alg:2}
\begin{algorithmic}[1]
\Given{Trajectory cost function $C(\cdot)$, kernel $k(\cdot,\cdot)$, number of particles $m$, number of iterations $iters$, inverse-temperature $\alpha$, step-size $\eta$}
\Input{Predicted state $\mathbf{x}_{t+1}$}
\Output{Safety-conditioned posterior $\tilde q_{t+1}(U)$}
\State \text{Initialize }$\tilde q_{t+1}(U)$\text{ and sample }$\{U\}_{i=1}^m\sim\tilde q_{t+1}(U)$
\For{$iter \leftarrow 1$ to $iters$}
    \For{$i\leftarrow 1$ to $m$ \textbf{in parallel}}
        \State $\nabla_{U^i}\log p_{t+1}(U^i|\mathcal{O}_{\tau^i}=1)$
        \Statex \hspace{3.0em}$=\nabla_{U^i}\log \mathbb{E}[\exp(-\alpha C(\tau^i))]+\nabla_{U^i}\log \tilde q_{t+1}(U^i)$
    \EndFor
    \For{$i\leftarrow 1$ to $m$ \textbf{in parallel}}
        \State $\Delta U^i$
        \Statex \hspace{3.0em}$\leftarrow \frac{1}{m}\sum\limits_{j=1}^m k(U^j,U^i)\nabla_{U^j}\log p_{t+1}(U^i|\mathcal{O}_{\tau^i}=1)$
        \Statex \hspace{4.0em} $+\nabla_{U^j}k(U^j,U^i)$
        \State $U^i\leftarrow U^i+\eta\Delta U^i$
    \EndFor
\EndFor
    \State $w^i \leftarrow \exp(-\alpha C(\tau^i))$
    \State $w^i \leftarrow \frac{w^i}{\sum_{j=1}^m w^j}$
    \State \Return $\tilde q_{t+1}(U)=\sum_{i=1}^m w^i\mathcal{N}(U|U^i,\Sigma)$
\end{algorithmic}
\end{algorithm}

To solve \eqref{eq:q_star_final}, we adapt SV-MPC~\cite{lambert2020stein}, which uses Stein variational gradient descent (SVGD)~\cite{liu2016stein} to construct a particle-based approximation of the posterior. 
The procedure is summarized in Algorithm~\ref{alg:2}. 

Specifically, the distribution $q\in Q$ is represented by $m$ particles $\{U^i\}^m_{i=1}$.
This particle-based representation allows the approximation to capture complex, potentially multimodal structures that cannot be represented by unimodal Gaussian distributions.
SVGD iteratively updates each particle according to line 6 of Algorithm~\ref{alg:2}.
The update consists of two components: the kernel $k(\cdot,\cdot)$ acts as a similarity weight that enables particles to share gradient information, while its derivative $\nabla_{U^j}k(U^j,U^i)$ introduces a repulsive force that prevents particles from collapsing to a single mode.
Together, these terms allow the particles to ``spread out'' and approximate the target posterior more effectively.
We use the radial basis function (RBF) kernel $k(U,U')=\exp(-\frac{1}{h}\Vert U-U'\Vert_2^2)$, where $h=\text{med}^2/\log m$ and {med} denotes the median pairwise particle distance. This mechanism is critical for achieving comprehensive coverage of the safe control space; by maintaining a diverse set of candidate trajectories, the filter can simultaneously capture multiple, potentially disjoint, safe modes.

At the beginning of Algorithm~\ref{alg:2}, $m$ particles are initialized from a zero-mean Gaussian distribution (line 1).
The algorithm then computes the gradient of the log posterior for each particle (lines 3--4),
which provides the update direction for minimizing the objective in~\eqref{eq:q_star_final}.
The particles are subsequently updated (lines 5--7).
The second term $\nabla_{U^j}k(U^j,U^i)$ in line 6 acts as a repulsive force between particles, enabling the approximation to capture multimodality.
After the updates, the particle weights $w_i$ are computed based on the trajectory cost (lines 8--9).
The algorithm then constructs a Gaussian mixture $\tilde q(U)=\sum^m_{i=1}w_i\mathcal{N}(U|U^i,\Sigma)$, where $\mathcal{N}(U|U^i,\Sigma)$ denotes a Gaussian density with mean $U^i$ and covariance $\Sigma$ (line 10). 
The resulting $\tilde q$ is used as the sampling distribution in Algorithm~\ref{alg:1}.

\subsection{Safety Guarantee} \label{subsec_safety-guarantee}

We now establish the safety guarantee of the proposed filter. 

\begin{proposition} \label{prop:1}
    Let $\mathbf{x}_t$ be the current state.
    If $\mathbf{u}_{t}^\mathrm{safe}=\mathbf{u}_{t}^\mathrm{nom}$, then the proposed filter guarantees safety at time $t$.
\end{proposition}
\begin{proof}
    If $\mathbf{u}_{t}^\mathrm{safe}=\mathbf{u}_{t}^\mathrm{nom}$, then by Algorithm~\ref{alg:1}, the filter has stored a control sequence $U_{t+1}^\mathrm{safe}$ with negative cost.
    By \eqref{eq:failure_set} and \eqref{eq:cost}, the trajectory induced by $U_{t+1}^\mathrm{safe}$ from $\mathbf{x}_{t+1}$ is safe.
    Therefore, the proposed filter guarantees safety at time $t$.
\end{proof}


However, this argument does not directly apply when $\mathbf{u}_{t}^\mathrm{safe}\ne\mathbf{u}_{t}^\mathrm{nom}$, i.e., when the nominal input is overridden.
In this case, the filter applies the first input of $U_{t}^\mathrm{safe}$ and constructs $U_{t+1}^\mathrm{safe}$ by shifting $U_{t}^\mathrm{safe}$ by one step, as in line 9 of Algorithm~\ref{alg:1}.
To guarantee safety at the next timestep, we must ensure that the shifted sequence can be safely extended.
To this end, we define a \textit{safe control invariant set}.


\begin{definition}[Safe control invariant set] \label{def:1}
    A set $\mathcal{C}\subseteq\mathcal{X}$ is a \textit{safe control invariant set} if $\mathbf{x}\notin\mathcal{L},\;\forall\mathbf{x}\in\mathcal{C}$, and for every $\mathbf{x}\in\mathcal{C}$, there exists a control input $\mathbf{u}^{\mathrm{inv}}(\mathbf{x})\in\mathcal{U}$ such that $f(\mathbf{x},\mathbf{u}^{\mathrm{inv}})\in\mathcal{C}.$
\end{definition}

For many robotic systems, a decelerating sequence can drive the system to a known safe control invariant set.
In the single-robot experiment in Section~\ref{subsec_dubins}, candidate backup sequences are sampled from control sequences that bring the robot to a complete stop at the terminal timestep $H$.
Once the robot has stopped, the shifted sequence is extended by appending a terminal input that maintains the stopped state, so that a safe control sequence exists at all subsequent timesteps.
This motivates Proposition~\ref{prop:2}.

\begin{proposition} \label{prop:2}
    Let the sampling space be restricted to control sequences whose terminal state lies in $\mathcal{C}$. Suppose there exists a known invariance control law $\mathbf{u}^{\mathrm{inv}}(\mathbf{x})$ such that for any $\mathbf{x} \in \mathcal{C}$, the next state remains in $\mathcal{C}$.
    
    If the filter does not intervene at $t=0$, i.e., $\mathbf{u}_0^{\mathrm{safe}}=\mathbf{u}_0^{\mathrm{nom}}$, then the filter guarantees safety at all timestep $t\ge0$.
\end{proposition}
\begin{proof}
    We prove the claim by induction on $t$.
    By Proposition~\ref{prop:1}, the condition $\mathbf{u}_0^{\mathrm{safe}}=\mathbf{u}_0^{\mathrm{nom}}$ implies that the filter guarantees safety at $t=0$. Furthermore, the corresponding safe control sequence $U_1^{\mathrm{safe}}$ steers the terminal state into $\mathcal{C}$ by the construction of the sampling space.
    
    For the induction step, assume that at time $t$, the filter guarantees safety with a safe control sequence $U_{t+1}^\mathrm{safe}=(\bar{\mathbf{u}}_{t+1},\cdots, \bar{\mathbf{u}}_{t+H})$ whose terminal state $\mathbf{x}_{t+H+1}$ lies in $\mathcal{C}$.
    
    At time $t+1$, if $\mathbf{u}_{t+1}^\mathrm{safe}=\mathbf{u}_{t+1}^\mathrm{nom}$, by Proposition~\ref{prop:1}, the filter guarantees safety at time $t+1$ with a safe control sequence whose terminal state lies in $\mathcal{C}$ due to the sampling space. Otherwise, the filter overrides $\mathbf{u}_{t+1}^\mathrm{nom}$ and applies $\mathbf{u}_{t+1}^{\mathrm{safe}}=\bar{\mathbf{u}}_{t+1}$. 
    Since the terminal state induced by $U_{t+1}^\mathrm{safe}$ lies in $\mathcal{C}$, we can construct a candidate safe sequence $U_{t+2}^{\mathrm{safe}}=(\bar{\mathbf{u}}_{t+2},\cdots, \bar{\mathbf{u}}_{t+H}, \mathbf{u}^{\mathrm{inv}}(\mathbf{x}_{t+H+1}))$.
    Hence, there exists a safe control sequence $U_{t+2}^\mathrm{safe}$ that keeps the terminal state in $\mathcal{C}$, and the filter guarantees safety at time $t+1$. 
\end{proof}


\begin{remark}
    To ensure all-time safety with a braking-based backup, the horizon $H$ must exceed the stopping time under maximum braking. This allows the filter to drive the system to a safe stationary state within the horizon.
\end{remark}

\begin{remark}
    From the CBF perspective, the cost function can also be defined as
    $$C(\tau^i_{t+1}):=\max_{k\in[0,H]}\left\{-l(\mathbf{x}^i_{t+k+1}) + (1-\gamma)l(\mathbf{x}^i_{t+k})\right\},$$
    where $\gamma\in(0,1)$ denotes the decay rate. 
    This formulation also guarantees safety, often yields smoother control actions, and is closely related to MPC-CBF in \cite{zeng2021safety}.
\end{remark}

\subsection{Restrictiveness}

The proposed safety filter intervenes when the minimum cost among the sampled trajectories is nonnegative. 
Ideally, the filter should intervene only when every possible control sequence from $\mathbf{x}_{t+1}=f(\mathbf{x}_t,\mathbf{u}^\mathrm{nom}_t)$ is unsafe, that is, when the minimum cost over the entire control space is nonnegative. 
Since this quantity cannot be evaluated exactly, the filter relies on $N$ sampled sequences to assess the safety of the nominal input.
We derive a probabilistic guarantee on the restrictiveness induced by this finite-sample approximation, based on the scenario approach~\cite{campi2009scenario}.

\begin{theorem}\label{thm:1}
    Given a restrictiveness parameter $\epsilon\in(0,1)$ and a confidence parameter $\beta\in(0,1)$, let the sample size $N$ satisfy 
    \begin{equation} \label{eq:N}
        N\ge\frac{2}{\epsilon}\left(\ln{\frac{1}{\beta}+1}\right).
    \end{equation}
    If the filter in Algorithm~\ref{alg:1} intervenes, i.e., $\mathbf{u}_t^\mathrm{safe}\ne \mathbf{u}_t^\mathrm{nom}$, then with probability at least $1-\beta$, the probability of sampling a safe control sequence $U_{t+1}$ from $\tilde q_{t+1}$ is at most $\epsilon$:
    \begin{equation} \label{eq:thm1}
        \Pr_{U_{t+1}\sim\tilde q_{t+1}}\!\left(\mathcal{O}_{\tau_{t+1}}=1\right)\le\epsilon.
    \end{equation}
\end{theorem}
\begin{proof}
    If the filter intervenes, then the minimum sampled cost satisfies $\min_i C(\tau^i_{t+1})\ge0$. 
    The problem of finding this minimum cost can be formulated as
    \begin{equation*}
        \begin{aligned}
            &\max_{g\in\mathbb{R}}\; g \\
            &\text{s.t.}\quad g\le C(\tau^i_{t+1}),\ i=1,\cdots,N,
        \end{aligned}
    \end{equation*}
    which is the same form as the scenario optimization problem in~\cite{campi2009scenario}. Therefore, by Theorem 1 in~\cite{campi2009scenario}, if the sample size $N$ satisfies \eqref{eq:N}, then $\Pr_{U_{t+1}\sim\tilde q_{t+1}}\!\left(g^*-C(\tau_{t+1})>0\right)\le\epsilon$.
    If the filter intervenes, then $g^*>0$. Hence,
    $$\Pr_{U_{t+1}\sim\tilde q_{t+1}}\!\left(C(\tau_{t+1})<0\right)\le \Pr_{U_{t+1}\sim\tilde q_{t+1}}\!\left(C(\tau_{t+1})<g^*\right) \le \epsilon.$$
    By \eqref{eq:safety} and \eqref{eq:cost}, this implies \eqref{eq:thm1}.
\end{proof}

\begin{remark}
    The probabilistic guarantee in Theorem~\ref{thm:1} holds for any sampling distribution $\tilde q_{t+1}$.
    However, the practical utility of the filter depends on the choice of this distribution. 
    A more accurate approximation of the safety-conditioned posterior, achieved in this paper through the multimodal coverage of SV-MPC, reduces unnecessary interventions by effectively exploring the safe control space.
\end{remark}

\section{EXAMPLES} \label{sec:exp}

In this section, we evaluate the proposed safety filter in two examples. 
We first consider a single-robot obstacle avoidance task to compare the sampling distributions induced by SV-MPC and CEM-based safety filters~\cite{park2025safety,feng2025words}, and to examine the restrictiveness of the resulting filters.
We then apply the proposed method to a multi-vehicle intersection scenario using a reinforcement learning-based nominal controller.

\subsection{Single-Robot Obstacle Avoidance} \label{subsec_dubins}

We demonstrate the proposed safety filter on a system with the following dynamics:
\begin{equation*}
    \begin{aligned}
        x_{t+1} &= x_t + v \cos(\theta_t) \Delta t,&&
        y_{t+1} = y_t + v \sin(\theta_t) \Delta t,\\
        \theta_{t+1} &= \theta_t + \omega_t \Delta t,&&
        v_{t+1} = v_t + a_t \Delta t.
    \end{aligned}
\end{equation*}
where $x$ and $y$ denote the robot position, $\theta$ is the heading angle, and $v\in[0,1]$ is the speed. The control inputs are the angular velocity $\omega\in[-1.5,1.5]$ and the acceleration $a\in[-1,1]$, with a timestep of $\Delta t=0.1$.

\begin{figure}[thpb]
  \centering
  \includegraphics[width=\columnwidth]{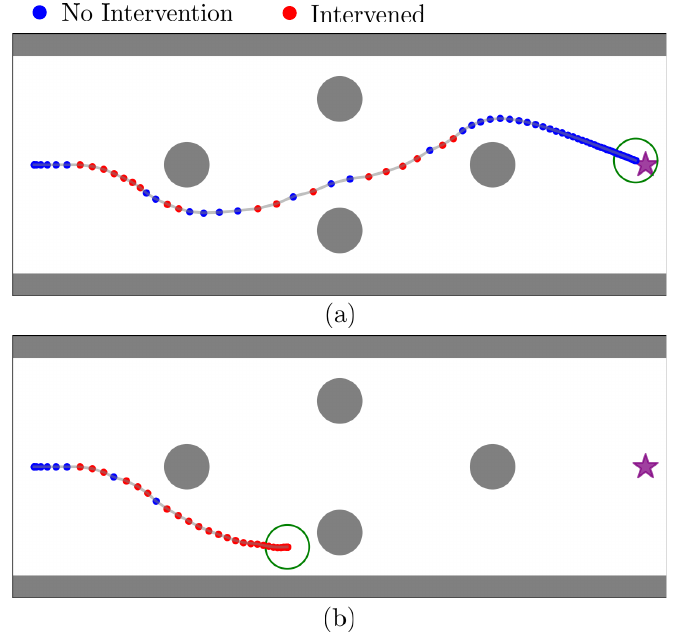}
  \caption{Comparison of closed-loop trajectories obtained under filters based on (a) SV-MPC and (b) CEM~\cite{park2025safety,feng2025words}.}
  \label{fig:1}
\end{figure}

As shown in Fig.~\ref{fig:1}, the robot is tasked with reaching the goal position, marked by a purple star, while avoiding the gray obstacles.
The robot is depicted as a hollow green circle.
To evaluate safety, we define the level function as
$$l(\mathbf{x})=1000\left(d_{\min}-(r_{\mathrm{robot}}+r_{\mathrm{obs}})\right),$$
which is used to compute the cost $C(\tau^i)$. Here, $d_{\min}$ denotes the minimum distance from the robot position to the obstacle centers, and $r_{\mathrm{robot}}=0.1$ and $r_{\mathrm{obs}}=0.1$ are the radii of the robot and the obstacles, respectively.
The minimum-distance operation leads to a nonsmooth and nonconvex level function, which motivates the use of the proposed sampling-based filter.
We apply the filter with $N=757$ and $H=20$. 
The sample size $N=757$ is determined from \eqref{eq:N} with $\beta=10^{-16}$ and $\epsilon=0.1$. 
For the SV-MPC procedure, we set $m=12$, $iters=5$, $\alpha=0.1$, and $\eta=0.25$.

Fig.~\ref{fig:1} illustrates the experimental results. 
Trajectories are shown in red when the filter intervenes and in blue otherwise. 
As the nominal controller, we use an MPC that seeks only to reach the goal without considering safety constraints.
As shown in Fig.~\ref{fig:1}(a), when the filter uses SV-MPC to construct the sampling distribution, the robot successfully reaches the goal along a safe trajectory. 
In contrast, when the filter uses CEM, the robot becomes stuck in a deadlock and fails to reach the goal, as shown in Fig.~\ref{fig:1}(b).

\begin{figure}[thpb]
  \centering
  \includegraphics[width=\columnwidth]{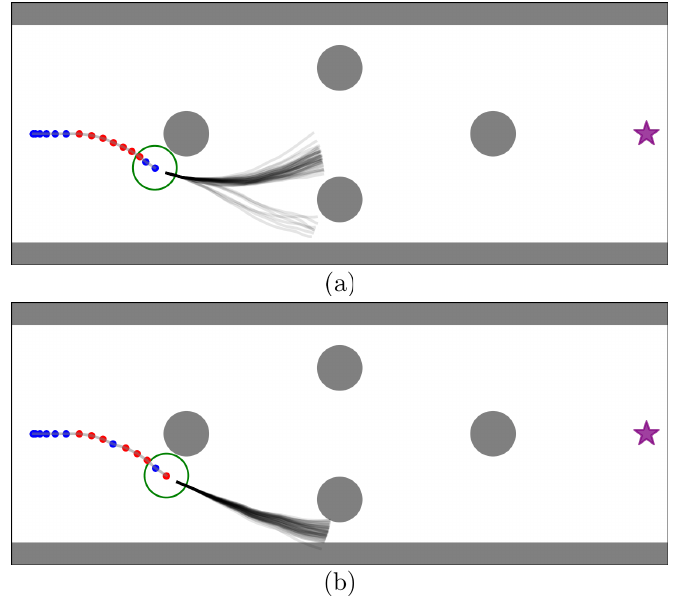}
  \caption{Comparison of samples generated by the filter using (a) SV-MPC and (b) CEM~\cite{park2025safety,feng2025words} to construct the sampling distribution.}
  \label{fig:2}
\end{figure}

In Fig.~\ref{fig:2}, we visualize the samples generated by the two filters to better understand this difference.
Fig.~\ref{fig:2}(a) shows samples generated using SV-MPC, which captures the multimodal structure of the safety-conditioned posterior and thereby enables broader exploration of safe trajectories.
Fig.~\ref{fig:2}(b) shows samples generated using a CEM-based distribution constructed with five iterations, which is the same number of iterations used in SV-MPC.
This procedure is similar to the distribution update in Algorithm 1 of~\cite{feng2025words}.
Because this distribution is unimodal Gaussian, the resulting samples concentrate around a single mode, leading to limited diversity.
As a result, depending on initialization and sampling, the distribution may converge to an unfavorable direction and become stuck in the deadlock shown in Fig.~\ref{fig:1}(b).

\begin{table}
\caption{Comparison of Filter Restrictiveness} \label{tab:1}
\begin{center}
\begin{tabular}{|c|c|c|>{\centering\arraybackslash}p{1cm}|c|c|}
\hline
\multicolumn{2}{|c|}{Parameters} & \multicolumn{2}{c|}{Num. Intervened} & \multicolumn{2}{c|}{Max. Safe Sample Rate}\\
\hline
$\epsilon$ &$N$ & SV-MPC & CEM & SV-MPC & CEM\\
\hline
0.2 & 379 & 4717 & 5794 & $1.4\times10^{-2}$ & $1.0\times10^{-2}$\\
\hline
0.1 & 757 & 4672 & 5730 & $3.9\times10^{-3}$ & $3.6\times10^{-3}$\\
\hline
0.01 & 7569 & 4549 & 5601 & $8.4\times10^{-4}$ & $4.3\times10^{-4}$\\
\hline
\end{tabular}
\end{center}
\end{table}

We also examine the restrictiveness of the filters.
To this end, we test whether the filter intervenes for all feasible $(x,y)$ positions in the same environment as in Fig.~\ref{fig:1} and Fig.~\ref{fig:2}, while fixing $v=1$ and $\theta=0$.
The region is discretized on a grid with resolution $0.01\times0.01$, and the number of states where the filter intervenes is reported as \textbf{Num. Intervened} in TABLE~\ref{tab:1}.
For each intervened state, we then draw an additional 100,000 control sequences from the same sampling distribution used by the filter and compute the fraction of safe samples, i.e., samples with negative cost.
The maximum of these fractions over all intervened states is reported as the \textbf{Max. Safe Sample Rate}.
This analysis is performed with $H=20$ and $N=379/757/7569$, which correspond to $\epsilon=0.2/0.1/0.01$, respectively, when $\beta=10^{-16}$. 
The results are summarized in TABLE~\ref{tab:1}.

As shown in TABLE~\ref{tab:1}, the filter based on SV-MPC intervenes less frequently than the filter based on CEM, indicating that SV-MPC better approximates the posterior.
In both methods, the number of interventions decreases as the sample size $N$ increases.
The safe sample rates are smaller than the corresponding restrictiveness parameter $\epsilon$ in all cases,
consistent with the bound $\epsilon$ in Theorem~\ref{thm:1}.

\subsection{Multi-Vehicle Intersection}

This subsection demonstrates the proposed safety filter in a multi-vehicle scenario.
We consider an unsignalized intersection consisting of a single-lane four-way crossing with multiple vehicles as shown in Fig.~\ref{fig:3}. 
Each vehicle is randomly initialized with an entry time, entry location, and destination. To model the longitudinal motion of each vehicle, we employ double integrator dynamics.
The state of vehicle $i$ at time $t$ is defined as
$s_{i,t}:=(x_{i,t}, y_{i,t}, v_{i,t}),$
where $(x,y)$ denotes the position and $v\in[0,10]$ denotes the velocity. 
The overall system state is formed by concatenating the states of all vehicles, i.e., $(s_{1,t}, s_{2,t},s_{3,t},\cdots)$.
The control input for each vehicle is defined as the longitudinal acceleration $a_{i,t}\in[-1.5,1.5]$. 
The paths of all vehicles toward their destinations are predetermined, and the control action affects only their longitudinal motion along these paths. 

As the nominal controller, we adopt the GPT-based Decision Transformer (DT) proposed in~\cite{lee2024gpt}. This offline reinforcement learning architecture treats the control task as a sequence-modeling problem, mapping historical states and desired returns-to-go to optimal actions. The DT is trained to minimize intersection traversal time while simultaneously reducing inter-vehicle collisions.

We apply the proposed safety filter to this nominal controller. 
The level function is defined as
\begin{equation*}
l(\mathbf{x}) =
\begin{cases}
100\,d, & d>0, \\
-100, & \text{otherwise},
\end{cases}
\end{equation*}
where $d=\min(d_{12},d_{23},d_{31})$ denotes the minimum pairwise distance between the vehicle polytopes. 
This piecewise-defined level function is discontinuous, which makes the setting well-suited to the proposed sampling-based approach.
The filter uses $N=757$ samples and a horizon $H=67$.
For SV-MPC, we set $m=12$, $iters=5$, $\alpha=0.1$ and $\eta=0.15$.

\begin{table}
\caption{Collision Rate and Filter Intervention Statistics} \label{tab:2}
\begin{center}
\begin{tabular}{|c|c|c|c|}
\hline
& Nominal & \multicolumn{2}{c|}{Filter Applied} \\
\hline
Num Vehicles & Collision Rate & Intervention Rate & Collision Rate \\
\hline
3 Vehicles & 0.04 & 0.08 & 0.00 \\
\hline
4 Vehicles & 0.12 & 0.21 & 0.00 \\
\hline
5 Vehicles & 0.15 & 0.38 & 0.00 \\
\hline
\end{tabular}
\end{center}
\end{table}

TABLE~\ref{tab:2} summarizes the experimental results for scenarios with 3, 4, and 5 vehicles. 
We randomly select 100 test cases not included in the training dataset to evaluate the effect of the safety filter. 
Without the filter, the nominal controller results in collisions in 4, 12, and 15 cases for 3-, 4-, and 5-vehicle settings, respectively. 
When the filter is applied, all collision cases are eliminated, although it also intervenes in some cases that would not have resulted in collisions.

\begin{figure}[thpb]
  \centering
  \includegraphics[width=\columnwidth]{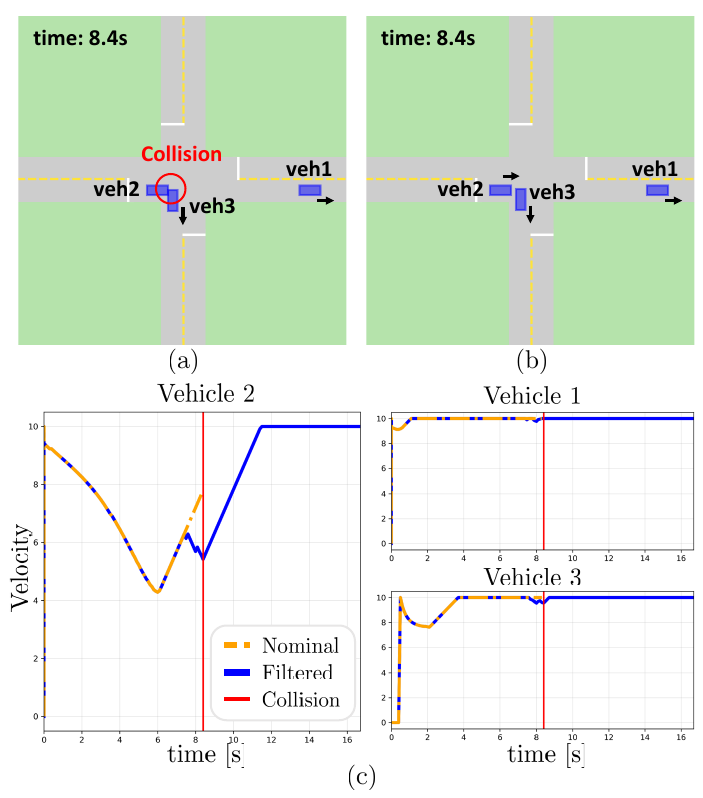}
  \caption{Visualization of a test case. (a) Vehicle 2 collides with vehicle 3 without the filter. (b) With the filter applied, vehicle 2 avoids the collision with vehicle 3. (c) Velocity profiles of the three vehicles.}
  \label{fig:3}
\end{figure}

We visualize an example test case, from the 3-vehicle setting, in which a collision occurs without the filter. 
As shown in Fig.~\ref{fig:3}(a), vehicle 2 collides with vehicle 3 without the filter.
With the filter applied, however, vehicle 2 slows down and avoids the collision, as shown in Fig.~\ref{fig:3}(b).
Fig.~\ref{fig:3}(c) shows the velocity profiles of all three vehicles.


\section{CONCLUSION} \label{sec:conclusion}

In this paper, we propose a sampling-based safety filter.
We establish a finite-sample probabilistic guarantee on restrictiveness via the scenario approach.
By adapting SV-MPC to approximate a safety-conditioned posterior over control sequences, the proposed filter can identify safe samples more effectively and thereby reduce unnecessary interventions.
Experimental results show that, with a sufficiently long horizon, the proposed method eliminates collisions entirely in both single- and multi-vehicle scenarios while being less restrictive than a CEM-based filter~\cite{park2025safety,feng2025words}.

The probabilistic restrictiveness guarantee in this work depends on the sampling distribution and therefore cannot serve as an absolute measure of filter restrictiveness.
In future work, we plan to relate the notion of filter restrictiveness directly to the backward reachable tube.


\bibliographystyle{IEEEtran}
\bibliography{main}

\end{document}